\documentclass[12pt]{article}

\usepackage{times}
\usepackage{fullpage}
\usepackage{amsfonts,amssymb,amsmath,amsthm}
\usepackage{latexsym}
\usepackage{gauss}
\usepackage{calc}
\usepackage{url}
\usepackage{thmtools}
\usepackage{thm-restate}
\usepackage{hyperref}
\usepackage[capitalise, nameinlink]{cleveref}
\usepackage{algorithm}
\usepackage{algorithmicx}
\usepackage{algpseudocode}
\usepackage{graphicx}
\usepackage{enumitem}
\newcommand{\R}{\mathbb{R}}

\newcommand{\E}{\mathbb{E}}
\newcommand{\zero}{\mathbf{0}}

\newcommand{\Acal}{\mathcal{A}}

\newcommand{\Dcal}{\mathcal{D}}

\newcommand{\Tcal}{\mathcal{T}}

\newcommand{\Recover}{\mathsf{RecoverOneFromAll}}

\newcommand{\lin}{\mathrm{lin}}

\newcommand{\polylog}{\mathrm{polylog}}

\newcommand{\ceil}[1]{\left\lceil #1 \right\rceil}

\def\01{\{0,1\}}

\newcommand{\makeset}{\mathrm{MakeSet}}

\newcommand{\roundedges}{\mathrm{RoundEdges}}
\newcommand{\elts}{\mathrm{Elts}}
\newcommand{\recover}{\mathrm{RecoverOneFromAll}}
\newcommand{\union}{\mathrm{Union}}
\newcommand{\findset}{\mathrm{FindSet}}
\newcommand{\rank}{\mathrm{rank}}
\newcommand{\outgoing}{\mathrm{Outgoing}}
\newcommand{\edgestoadd}{\mathrm{EdgesToAdd}}

\newcommand{\ith}{i^{\scriptsize \mbox{{\rm th}}}}
\newcommand{\jth}{j^{\scriptsize \mbox{{\rm th}}}}

\newcommand{\braket}[2]{\langle#1, #2\rangle}
\newcommand{\rk}{\mathrm{rk}}

\newtheorem{theorem}{Theorem}

\newtheorem{lemma}[theorem]{Lemma}
\newtheorem{corollary}[theorem]{Corollary}
\newtheorem{proposition}[theorem]{Proposition}
\newtheorem{remark}[theorem]{Remark}

\theoremstyle{definition}

\newtheorem{definition}[theorem]{Definition}

\newcommand{\mincutcert}{\textrm{mincut-cert}}
\newcommand{\concert}{\textrm{con-cert}}
\newcommand{\cutrk}{\textrm{-cut-rank}}

\newcommand{\taucerta}{\textrm{at-least-}\tau\textrm{-cert}}

\newcommand{\mincut}{\mathrm{MINCUT}_n}
\newcommand{\con}{\mathrm{CON}_n}

\newcommand{\Bcal}{\mathcal{B}}
\newcommand{\Temp}{\mathrm{Temp}}

\begin{document}
\title{On the query complexity of connectivity with global queries}
\author{Arinta Auza\thanks{Centre for Quantum Software and Information, University of Technology Sydney.} \and Troy Lee\thanks{Centre for Quantum Software and Information, University of Technology Sydney. Email: troyjlee@gmail.com}}
\date{}
\maketitle

\begin{abstract}
We study the query complexity of determining if a graph is connected with global queries.  
The first model we look at is matrix-vector multiplication queries to the 
adjacency matrix.  Here, for an $n$-vertex graph with adjacency matrix $A$, one can query a vector $x \in \{0,1\}^n$ and receive the answer $Ax$.  
We give a randomized algorithm that can output a spanning forest of a weighted graph with constant probability after 
$O(\log^4(n))$ matrix-vector multiplication queries to the adjacency matrix.  This complements a result of Sun et al.\ (ICALP 2019) that gives a randomized algorithm 
that can output a spanning forest of a graph after $O(\log^4(n))$ matrix-vector multiplication queries to the signed vertex-edge incidence matrix of the graph.  
As an application, we show that a quantum algorithm can output a spanning forest of an unweighted 
graph after $O(\log^5(n))$ cut queries, improving and simplifying a result of Lee, Santha, and Zhang (SODA 2021), which gave the bound $O(\log^8(n))$.  

In the second part of the paper, we turn to showing lower bounds on the linear query complexity of determining if a graph is connected.  If $w$ is the 
weight vector of a graph (viewed as an $\binom{n}{2}$ dimensional vector), in a linear query one can query any vector $z \in \R^{n \choose 2}$ and receive the answer $\langle z, w\rangle$.  We show that a zero-error randomized 
algorithm must make $\Omega(n)$ linear queries in expectation to solve connectivity.  As far as we are aware, this is the first lower bound of any kind on the unrestricted linear query complexity of 
connectivity.  We show this lower bound by looking at the linear query \emph{certificate complexity} of connectivity, and characterize this certificate complexity in a linear algebraic fashion.
\end{abstract}

\section{Introduction}
Determining whether or not a graph is connected is a fundamental problem of computer science, which has been studied in many different computational models \cite{Wig92}.  
In this paper, we study query algorithms for connectivity that make use of \emph{global} queries.  Traditionally, the most common query models for studying
graph problems have been \emph{local} query models where one can ask queries of the form (i) Is there an edge between vertices $u$ and $v$? (ii) what is the 
degree of vertex $v$? (iii) what is the $\ith$ neighbor of $v$?  For example, in the adjacency matrix model one can ask queries of type~(i), and in the 
adjacency list model one can ask queries of type (ii) and (iii).  We call these local query models as each query has information about a 
single vertex or single edge slot.  It is known that the randomized query complexity 
of connectivity on a graph with $n$ vertices and $m$ edges is $\Theta(n^2)$ in the adjacency matrix model and $\Theta(m)$ in the adjacency array model \cite{DHHM06, BGMP20}.

On the other hand, a single global query can aggregate information about many vertices and edge slots.  There has recently been a lot of interest 
in the complexity of graph problems with various global queries: matrix-vector multiplication queries to the adjacency or (signed) vertex-edge incidence matrix \cite{SWYZ19,CHL21}, cut queries \cite{RSW18,ACK21,LSZ21}, and bipartite 
independent set (BIS) queries \cite{BHRRS20}, to name a few.  If $A$ is the adjacency matrix of an $n$-vertex graph, in a matrix-vector multiplication query to the adjacency matrix 
one can ask a vector $x \in \{0,1\}^n$ and 
receive the answer $Ax \in \{0,1\}^n$, in a cut query on question $x$ one receives $x^TA(1-x)$, and in a BIS query on question $x,y \in \{0,1\}^n$, where $x,y$ are \emph{disjoint}, the 
answer is $1$ if $x^T A y$ is positive and $0$ otherwise.  Motivation to study each of these models comes from different sources: matrix-vector multiplication queries have applications to streaming algorithms \cite{AGM12}, 
cut queries are motivated by connections to submodular function minimization, and bipartite independent set 
queries have been studied in connection with reductions between counting and decision problems \cite{DL21}.

Let us take the example of matrix-vector multiplication queries to the adjacency matrix, which is the primary model we study.  This is a very powerful model, as one learns an 
$n$-dimensional vector with a single query.  On the other hand, this model can also lead to very efficient algorithms: we show that 
a randomized algorithm making $O(\log^4(n))$ many matrix-vector queries to the adjacency matrix of $G$ can output a spanning forest of $G$ with constant probability (see \cref{cor:mv}).  

This result answers a natural question left open by Sun et al.\ \cite{SWYZ19}.  They introduce the study of the complexity of graph problems with matrix-vector multiplication queries, and 
ask if the particular representation of a graph by a matrix makes a difference in the complexity of solving certain problems.  Besides the adjacency matrix, another natural 
representation of an $n$-vertex simple graph $G = (V,E)$ is the signed vertex-edge incidence matrix $A_{\pm} \in \{-1,0,1\}^{n \times \binom{n}{2}}$.  The rows of $A_{\pm}$ are labeled by vertices and the columns are labeled by 
elements of $V^{(2)}$, the set of all 2-element subsets of $V$.  Fix an ordering of the vertices.  Then 
\[
A_{\pm}(u, \{v,w\}) = 
\begin{cases}
0 & \text{ if }  \{v,w\} \not \in E \text{ or } u \not \in \{v,w\} \\
1 & \text{ if } \{v,w\} \in E  \text{ and } u \text{ is the smallest element in } \{v,w\} \\
-1 & \text{ if } \{v,w\} \in E \text{ and } u \text{ is the largest element in } \{v,w\} \enspace .

\end{cases}
\]
Sun et al.\ observe that a beautiful sketching algorithm of Ahn, Guha, and McGregor \cite{AGM12} to compute a spanning forest of a graph $G$ via $O(n \log^3(n))$ 
non-adaptive linear measurements also gives a non-adaptive randomized query algorithm that finds a spanning forest after $O(\log^4(n))$ matrix-vector 
multiplication queries to $A_{\pm}$.  

For a \emph{bipartite} graph $G$ with bipartition $V_1, V_2$, the 
\emph{bipartite adjacency matrix} is the submatrix of the adjacency matrix where rows are restricted to $V_1$ and columns to $V_2$.  Sun et al.\ show that 
$\Omega(n/\log n)$ matrix-vector multiplication queries to the \emph{bipartite} adjacency matrix (where multiplication of the vector is on the right) can be needed to 
determine if a bipartite graph is connected or not.  However, the family of instances they give can be solved by a single query to the full adjacency matrix of 
the graph (or a single matrix-vector multiplication query \emph{on the left} to the bipartite adjacency matrix).  The arguably more natural question of comparing 
the complexity of matrix-vector multiplication queries in the signed vertex-edge incidence matrix versus the adjacency matrix model was left open.  Our results mean that
connectivity cannot show a large separation between these models.

While the matrix-vector multiplication query model is very powerful, as one can also hope for very efficient algorithms it still has interesting applications to weaker models, 
like cut queries.  As an example, it is not hard to see that each entry of $A_{\pm} x$ can be computed with a constant number of cut queries, and therefore
a matrix-vector multiplication query to $A_{\pm}$ can be simulated by $O(n)$ cut queries.  Thus the randomized non-adaptive $O(\log^4(n))$ matrix-vector 
query algorithm for connectivity in this model implies a randomized non-adaptive $O(n \log^4(n))$ algorithm for connectivity in the cut query model, which 
is state-of-the-art for the cut query model \cite{ACK21}.

We look at applications of matrix-vector multiplication algorithms to \emph{quantum} algorithms using cut and BIS queries.  Recent work of Lee, Santha, and 
Zhang showed that a quantum algorithm can output a spanning forest of a graph with high probability after $O(\log^8(n))$ cut queries \cite{LSZ21}.  This is in contrast 
to the randomized case where $\Omega(n/\log n)$ cut queries can be required to determine if a graph is connected \cite{BFS86}.  A key observation 
made in \cite{LSZ21} is that a quantum algorithm with cut queries can efficiently simulate a restricted version of a matrix-vector multiplication query to the adjacency matrix.  
Namely, if $A$ is the adjacency matrix of an $n$-vertex simple graph, a quantum algorithm can compute $Ax \circ (1-x)$ for $x \in \{0,1\}^n$ with $O(\log n)$ cut queries.
Here $\circ$ denotes the Hadamard or entrywise product.  In other words, the quantum algorithm can efficiently compute the entries of $Ax$ where $x$ is $0$.  

We quantitatively improve the Lee et al.\ result to show that a quantum algorithm can output a spanning tree of a simple graph with constant probability after $O(\log^5(n))$ cut queries (\cref{thm:qcut}).  
We do this by showing that the aforementioned $O(\log^4(n))$ adjacency-matrix-vector multiplication query algorithm to compute a spanning forest also works with the 
more restrictive $Ax \circ (1-x)$ queries.  In addition to quantitatively improving the result of \cite{LSZ21}, 
this gives a much shorter proof and nicely separates the algorithm into a quantum part, simulating $Ax \circ (1-x)$ queries, and a classical randomized algorithm using $Ax \circ (1-x)$ queries.  

This modular reasoning allows us to extend the argument to other models as well.  We also show that a quantum algorithm can compute a spanning 
forest with constant probability after $\tilde O(\sqrt{n})$ BIS queries.  This is done by simulating an $Ax \circ (1-x)$ by a quantum algorithm making $O(\sqrt{n})$ 
BIS queries (\cref{thm:con_master}).  

In the second part of this work we turn to lower bounds for connectivity, and here we focus on the model of \emph{linear queries}.  Let $G$ be an $n$-vertex weighted and undirected 
graph on vertex set $V$, and let $w : V^{(2)} \rightarrow \R_{\ge 0}$ be its weight function, which assigns a non-negative (possibly zero) weight to each 
two-element subset of $V$.  We will view $w$ as an $\binom{n}{2}$ dimensional vector.  In the linear query model, one can query any $x \in \R^{\binom{n}{2}}$ and receive 
the answer $\braket{w}{x}$.  As far as we are aware, no lower bound at all was known on the query complexity of connectivity in this unrestricted linear query model.  The reason is that 
the number of bits in a query answer is unbounded in the linear query model, which causes problems for any kind of information-theoretic lower bound technique.
Most lower bounds for connectivity go through the model of communication complexity and restrict the set of ``hard" instances to simple graphs.  However, given that a graph 
is simple, one can learn the entire graph with a single linear query by querying the vector of powers of $2$ of the appropriate dimension.  One can extend this argument to show that 
in fact any problem where the set of input instances is finite can be solved with a single linear query.
This situation is unusual in the typically discrete world of query complexity, and few techniques have been developed that work to lower bound the linear query model.  
Two examples that we are aware of are a linear query lower bound on the complexity of computing the minimum cut of a graph using the \emph{cut dimension} method \cite{GPRW20}, and its extension 
to the ``$\ell_1$-approximate cut dimension" by \cite{LLSZ21}, 
and a linear query lower bound for a problem called the 
\emph{single element recovery problem} by \cite{ACK21}.  

We show that any zero-error randomized algorithm that correctly solves connectivity must make $\Omega(n)$ linear queries in expectation (\cref{cor:lin_low}).  
We do this by building on the $\ell_1$-approximate cut dimension approach.  The $\ell_1$-approximate cut dimension was originally applied to show lower bounds on the 
deterministic linear query complexity of the minimum cut problem.  We show that this method actually characterizes, up to an additive $+1$, 
the linear query \emph{certificate complexity} of the minimum cut problem.  Certificate complexity is a well-known lower bound technique for query complexity.
The certificate complexity of a function $f$ on input $x$ is the minimum number of queries $q_1, \ldots, q_k$ 
needed such that any input $y$ which agrees with $x$ on these queries $q_1(x) = q_1(y), \ldots, q_k(x) = q_k(y)$ must also satisfy $f(x) = f(y)$.  The maximum certificate complexity over all inputs $x$ is 
a lower bound on the query complexity of deterministic and even zero-error randomized algorithms for $f$.  

In our context, the linear query certificate complexity of \emph{connectivity} on a graph $G = (V,w)$ is the minimum $k$ for which there are 
queries $q_1, \ldots, q_k$ such that 
for any graph $G' = (V,w')$, if $\braket{w}{q_i} = \braket{w'}{q_i}$ for all $i=1, \ldots, k$ then $G'$ is connected iff $G$ is.  
We adapt the $\ell_1$-approximate cut dimension approach to give a linear algebraic characterization of the linear query certificate complexity of connectivity.
We then show that the linear query certificate complexity for 
connectivity of the simple $n$-vertex cycle is $\Omega(n)$, giving the $\Omega(n)$ linear query lower bound for zero-error algorithms solving connectivity.

\section{Preliminaries}
For a natural number $n$ we let $[n] = \{1, \ldots, n\}$.
We represent an undirected weighted graph as a pair $G = (V,w)$, where $V$ is the set of vertices,
the set of edge slots $V^{(2)}$ is the
set of subsets of $V$ with cardinality 2, and the weight function
$w : V^{(2)} \rightarrow \R $ is non-negative. The set of edges of $G$
is defined as $E(G) = \{e \in V^{(2)} : w(e) > 0\}$.  When all edges of $G$ have weight $1$ we use the notation $G = (V,E)$, where $E \subseteq V^{(2)}$ is the 
set of edges of $G$.

\section{Basic spanning forest algorithm}
In this section, we give a template spanning forest algorithm based on Bor\r{u}vka's algorithm \cite{NMN01}.  In Bor\r{u}vka's algorithm on input a graph $G = (V,w)$, one maintains the invariant of having 
a partition $\{S_1, \ldots, S_k\}$ of $V$ and a spanning tree for each $S_i$.  At the start of the algorithm this partition is simply $V$ itself.  In a generic round of the algorithm, the goal 
is to find one outgoing edge from each $S_i$ which has one.  One then selects a subset $H$ of these edges which does not create a cycle among $S_1, \ldots, S_k$, and 
merges sets connected by edges from $H$ while updating the spanning trees for each set accordingly.  If $q$ of the $k$ sets have an outgoing edge, then the cycle free 
subset $H$ will satisfy $|H| \ge q/2$.  Every edge from $H$ decreases the number of sets by at least $1$, so the number of sets after this round is at most $k - q/2$.  
In this way we see that $k$ minus the number of connected components of $G$ at least halves with every round, and the algorithm terminates with a spanning forest of $G$ after $O(\log n)$ rounds.

The main work of a round of this algorithm is the problem of finding an outgoing edge from each $S_i$.  We abstract out performing this task 
by a primitive called $\Recover$.  Given oracle access to a non-negative matrix $A \in \R^{n \times n}$, a set of rows $R$, a set of columns $S$, and an error parameter $\delta$, with 
probability at least $1-\delta$ $\Recover$ does the following: for every $i \in R$ for which there is a $j \in S$ with $A(i,j) > 0$ it outputs a pair $(i,j')$ with $A(i,j') > 0$. 
We record the input/output behavior of $\Recover$ in this code header.
\begin{algorithm}[H]
\caption{$\Recover[A](R,S, \delta)$}
\label{alg:mindeg}
 \hspace*{\algorithmicindent} \textbf{Input:} Oracle access to a non-negative matrix $A \in \R^{m \times n}$, subsets $R \subseteq [m], S \subseteq [n]$, and an error parameter $\delta$. \\
 \hspace*{\algorithmicindent} \textbf{Output:} With probability at least $1-\delta$ output a set $\Tcal \subseteq R \times S$ such that for every $i \in R$ for which $\exists j \in S$ with $A(i,j) > 0$ there is a 
 pair $(i,t) \in \Tcal$ with $A(i,t) > 0$.
\end{algorithm}

Assadi, Chakrabarty, and Khanna \cite{ACK21} study a similar primitive for connectivity called \emph{single element recovery}.  In this problem one is given a non-zero non-negative vector $x \in \R^N$ and the 
goal is to find a coordinate $j$ with $x_j > 0$.  $\Recover$ is exactly the matrix version of this problem, where one wants to solve the single element recovery problem on every row of a matrix.  This version is 
more suitable for the matrix-vector multiplication query algorithms we study here where we want to implement a round of Bor\r{u}vka's algorithm in a parallel fashion.

In the subsequent sections we will see how $\Recover$ can be implemented in various global models.  First we analyze how many calls to $\Recover$ are needed to find a spanning forest.
\begin{algorithm}
\caption{FindSpanningForest}
\label{alg:generic_spanning}
 \hspace*{\algorithmicindent} \textbf{Input:} Subroutine RecoverOneFromAll$[A](R,S)$, where $A$ is the adjacency matrix of a weighted graph $G=(V,w)$. \\
 \hspace*{\algorithmicindent} \textbf{Output:} Spanning forest $F$ of $G$.
\begin{algorithmic}[1]
\State $\Bcal \gets \emptyset, F \gets \emptyset, \delta \gets  (300(\log(n)+10))^{-1}, T = 3 (\log(n) + 10)$
\For{$v \in V$}
  \State $\Bcal \gets \Bcal \cup \makeset(v)$ \Comment{Set up Union-Find data structure}
\EndFor
\For{$i = 1$ to $T$}
  \State $R \gets \emptyset, S \gets \emptyset$
  \For{$v \in \Bcal$}
      \State Flip a fair coin.  If heads $R \gets R \cup \elts(v)$, else $S \gets S \cup \elts(v)$.
   \EndFor 
   \State $\roundedges \gets \recover[A](R,S, \delta)$ \Comment{Note $R \cap S = \emptyset$}
  \State $\edgestoadd \gets \emptyset$ \Comment{$\edgestoadd$ will hold one outgoing edge from each set}
  \For{$v \in \Bcal$}  
    \State $\outgoing \gets \{e \in \roundedges: |e \cap \elts(v)| = 1\}$ 
    \If{$\outgoing \ne \emptyset$}
    	\State Add one edge from $\outgoing$ to $\edgestoadd$. 
    \EndIf
  \EndFor
  \State $F \gets F \cup \edgestoadd$
  \For{$\{u,v\} \in \edgestoadd$}
    \State $\union(u,v)$ \Comment{Merge sets connected by an edge of $\edgestoadd$}
  \EndFor
  \State $\Temp \gets \emptyset$
  \For{$v \in \Bcal$}
    \State $\Temp \gets \Temp \cup \findset(v)$ \Comment{Create updated partition}
  \EndFor   
  \State $\Bcal \gets \Temp$
\EndFor
\State Return $F$
\end{algorithmic}
\end{algorithm}

\begin{theorem}
\label{thm:main}
Let $G = (V,w)$ be a weighted graph with $n$ vertices, and $A$ its adjacency matrix.  There is a randomized algorithm that finds a spanning forest of $G$ with probability at least $49/50$ after making $3 (\log(n) + 10)$ calls to 
$\Recover[A](R,S,\delta)$ with error parameter $\delta = (300(\log(n)+10))^{-1}$.  Moreover, all calls to $\Recover$ involve sets $R,S \subseteq V$ that are disjoint.
\end{theorem}

\begin{proof}
The algorithm is given in \cref{alg:generic_spanning}.  We use a Union-Find data structure to maintain a partition of the vertices into sets 
which are connected.  This data structure has the operations $\makeset$ to create a set of the partition, $\findset$ which takes as argument 
a vertex and returns a representative of the set containing the vertex, and $\union$ which takes as arguments two vertices and merges their corresponding 
sets.  We also suppose the data structure supports the operation $\elts$ which, given a vertex $v \in V$, returns all the elements in the same set of the partition as $v$.

The algorithm makes one call to $\Recover$ in each iteration of the for loop.  There are $T =3 (\log(n) + 10)$ iterations of the for loop, giving the upper bound on the complexity.  Let us now show correctness.
We will first show correctness assuming the procedure $\Recover$ has no error.

In an arbitrary round of the algorithm we have a partition $B_1, \ldots, B_k$ of $V$ and a set of edges $F$ which consists of a spanning tree for each $B_i$.
For $i = 1, \ldots, k$ we color $B_i$ red or blue independently at random with equal probability.  All vertices in $B_i$ are given the color of $B_i$.
Let $R$ be the set of red vertices and $S$ the set of blue vertices.  We then call $\Recover[A](R,S,\delta)$.  Note that $R \cap S = \emptyset$, 
showing the ``moreover'' part of the theorem.  Assuming that $\Recover$ returns without error, this gives us an edge in the red-blue cut 
for every red vertex which has such an edge.  The set of all these edges is called $\roundedges$.  We then select from $\roundedges$ one edge incident to each red $B_i$ which has one, 
and let the set of these edges be $\edgestoadd$.  
This set of edges necessarily does not create a cycle among the sets $B_1, \ldots, B_k$ as we have chosen at most one edge from only the red $B_i$.  
We add $\edgestoadd$ to $F$ and merge the sets connected by these edges giving a new partition $B_1', \ldots, B_{k'}'$.  Assuming that $\Recover$ returns 
without error, so that all added edges are valid edges, and as we have added a cycle free set of edges, we maintain the invariant that $F$ contains a spanning 
tree for each set of the partition.  

To show correctness it remains to show that with constant probability at the end of the algorithm the partition $B_1, \ldots, B_k$ consists of connected components.
The key to this is to analyze how the number of sets in the partition decreases 
with each round.  Let $Q_i$ be a random variable denoting the number of sets of the partition that have an outgoing edge at the start of round $i$.  
Thus, for example $Q_1$ is equal to $n$ minus the number of isolated vertices with probability $1$.  In particular $\E[Q_1] \le n$.

Consider a generic round $i$ and say that at the start of this round $q_i$ many sets of the partition have an outgoing edge.  
In expectation $q_i/2$ of the sets with an outgoing edge are colored red, and each outgoing edge has probability $1/2$ of being in the red-blue cut.
Thus the expected size of $\edgestoadd$ is at least $q_i/4$.  As each edge in $\edgestoadd$ reduces the number of sets in the partition of the next round by at least one this means 
$\E[Q_{i+1} \mid Q_i = q_i] \le 3q_i/4$.  Therefore we have 
\begin{align*}
\E[Q_{i+1}] = \sum_{q_i} \E[Q_{i+1} \mid Q_i = q_i] \Pr[Q_i = q_i] &\le \frac{3}{4} \sum_{q_i} q_i  \Pr[Q_i = q_i] \\
&= \frac{3}{4} \E[Q_i] \enspace .
\end{align*}
From this it follows by induction that $\E[Q_i] \le (3/4)^{i-1} n$.  Taking $T = 3 (\log(n) + 10)$ we have $\E[Q_T] \le 1/100$, and therefore by Markov's inequality the algorithm will terminate 
with all sets being connected components except with probability at most $1/100$.

Finally, let us remove the assumption that $\Recover$ returns without error.  We have set the error parameter in $\Recover$ to be $\delta = (300(\log(n)+10))^{-1}$.  
As $\Recover$ is called $T$ times, by a union bound the probability it ever makes an error is at most $T \delta \le 1/100$.  Therefore adding this to our error bound, 
the algorithm will be correct with probability at least $49/50$.
\end{proof}

\section{Master Model}
Given \cref{alg:generic_spanning}, the task now becomes to implement $\Recover[A](R,S,\delta)$.  There is a natural approach to do this 
following the template of $\ell_0$ samplers \cite{JST11}.  Let $B = A(R,S)$ be the submatrix of interest, and assume for the moment that $A \in \{0,1\}^{n \times n}$ is Boolean.
\begin{enumerate}
  \item Estimate the number of ones in each row of $B$.
  \item Bucket together the rows whose estimate is in the range $(2^{i-1}, 2^i]$ into a set $G_i$.
  \item For each $G_i$, randomly sample $\Theta(n \log(n)/2^i)$ columns of $B$.  With high probability every row of $G_i$ will have at least one and at most $c \cdot \log(n)$ ones 
  in this sample, for a constant $c$.  
  \item For each $i$ and row in $G_i$ learn the $O(\log(n))$ ones in the sampled set.
\end{enumerate}

We now want to find a ``master query model'' that can efficiently implement this template, yet is sufficiently weak that simulations of this model by other global 
query models give non-trivial results.  To do this we take inspiration from combinatorial group testing.  In combinatorial group testing one is given \emph{OR query} access to a string $x \in \{0,1\}^n$.  
This means that one can query any subset $S \subseteq [n]$ and receive the answer $\vee_{i \in S} \; x_i$.  Both the problems of estimating the number of ones in $x$ we need for 
step~1 and learning a sparse string $x$ we need for step~4 have been extensively studied in the group testing setting.  A natural kind of query to implement group testing algorithms in the 
matrix setting is a Boolean matrix-vector multiplication query, that is for $A \in \{0,1\}^{n \times n}$ one can query $x \in \{0,1\}^n$ and receive the answer $y = A \vee x \in \{0,1\}^n$ 
where $y_i = \vee_j (A(i,j) \wedge x_j)$.  With this kind of query one can implement a \emph{non-adaptive} group testing algorithm in parallel on all rows of $A$.

We will define the master model to be a weakening of the Boolean matrix-vector multiplication query, in order to accomodate quantum cut queries later on.  In the master model one can query $x \in \{0,1\}^n$ 
and receive the answer $(A \vee x) \circ (1-x)$.  Here $\circ$ denotes the Hadamard or entrywise product of vectors.  In other words, in the 
master model one only learns the entries of $A \vee x$ in those coordinates where $x$ is zero.  This restriction allows one to efficiently simulate these queries 
by a quantum algorithm making cut queries.

If $A$ is non-Boolean then in the master model one can again query any $x \in \{0,1\}^n$ and receive the answer $[Ax]_{> 0} \circ (1-x)$, 
where $[y]_{> 0} \in \01^n$ is a Boolean vector whose $i^{th}$ entry is $1$ if $y_i > 0$ and $0$ otherwise.  The following proposition shows that 
it suffices to restrict our attention to Boolean matrices.

\begin{proposition}
\label{prop:boolean}
Let $A \in \R^{n \times n}$ be a non-negative matrix and $x \in \01^n$.  Then $[Ax]_{> 0} = [A]_{> 0} \vee x$.  
\end{proposition}

\begin{algorithm}
\caption{Implementation of $\Recover$ in the master model.}
\label{alg:generic_spanning}
 \hspace*{\algorithmicindent} \textbf{Input:} $(A \vee x) \circ (1-x)$ query access to a Boolean matrix $A \in \01^{n \times n}$, disjoint subsets $R,S \subseteq [n]$, 
 and an error parameter $\delta$ \\
 \hspace*{\algorithmicindent} \textbf{Output:} A set $\Tcal \subseteq R \times S$ such that for every $i \in R$ if there is a $j \in S$ with $A(i,j) =1$ then 
 there is a pair $(i,t) \in \Tcal$ with $A(i,t) = 1$, except with error probability at most $\delta$.
\begin{algorithmic}[1]
\State Estimate $|A(i,S)|$ for all $i \in R$ using \cref{lem:counting}.
\For{$i=1, \ldots, \ceil{\log n}$}
     \State Let $G_i$ be the set of rows of $R$ whose estimate is in the range $(2^{i-1}, 2^i]$
     \State Randomly sample $\min\{n, \ceil{32 |S| \ln(n)/2^i}\}$ elements from $S$ with replacement, and let the selected set be $H_i$
     \State Learn the submatrix $A(G_i, H_i)$
\EndFor
\end{algorithmic}
\end{algorithm}

The main result of this section is the following.
\begin{lemma}
\label{lem:recover}
Let $A \in \R^{n \times n}$ be a non-negative matrix.  $\Recover[A](R,S,1/n^2)$ on disjoint subsets $R,S \subseteq [n]$ can be implemented with
$O(\log(n)^3)$ many $[Av]_{>0} \circ (1-v)$ queries.
\end{lemma}

At a high level, this lemma will follow the 4 steps given at the beginning of the section.  Step~2 does not require any queries.
Step~3 follows easily by a Chernoff bound, as encapsulated in the following lemma.
\begin{lemma}[{\cite[Lemma 8]{LSZ21}}] \label{lem:sample}
Let $x^{(1)}, \ldots, x^{(k)} \in \{ 0,1\}^\ell$ be such that $\frac{t}{8} \leq |x^{(i)}| \leq 2t$ for all $i \in \{1, \ldots, k\}$. For $\delta > 0$, sample with replacement $\frac{8\ell \ln(k/\delta)}{t}$ elements of $\{1, \ldots, n\}$, and call the resulting set $R$. Then
\begin{itemize}
    \item $\Pr_R[\exists i \in \{1, \ldots, k\}: |x^{i}(R)| = 0]\leq \delta$
    \item $\Pr_R[\exists i \in \{1, \ldots, k\}: |x^{i}(R)| > 64 \ln{(k/\delta)}]\leq \delta$
\end{itemize}
\end{lemma}

The interesting part are steps~1 and~4, both of which will be done using non-adaptive group testing algorithms.  Towards step~1 we make the following definition.
\begin{definition}[Good estimate]
\label{def:good}
We say that $b \in \R^k$ is a good estimate of $c \in \R^k$ iff $b(i) \le c(i) \le 2 b(i)$ for all $i \in \{1, \ldots, k\}$.  
\end{definition}

The next theorem gives a randomized non-adaptive group testing algorithm for estimating the number of ones in a string $x$.  
\begin{theorem}[{\cite[Theorem 4]{Bshouty19}}]
\label{lem:counting}
Let $n$ be a positive integer and $0 < \delta < 1$ an error parameter.  There is a non-adaptive randomized algorithm that on input 
$x \in \{0,1\}^n$ outputs a good estimate of $|x|$ with probability at least $1-\delta$ after $O(\log(1/\delta)\log(n))$ many OR queries to $x$.
\end{theorem}

The problem of learning a string $x$ by means of OR queries is the central problem of combinatorial group testing and has been extensively studied \cite{DH99}.
It is known that non-adaptive \emph{deterministic} group testing algorithms must make $\Omega\left(\frac{d^2 \log(n)}{\log(d)}\right)$  queries in order 
to learn a string with at most $d$ ones \cite{DR82,Rus94}.  However, there are non-adaptive \emph{randomized} group testing algorithms that learn 
any $x \in \{0,1\}^n$ with $|x| \le d$ with probability at least $1-\delta$ after $O(d(\log(n) + \log(1/\delta)))$ queries.  This is what we will use.

\begin{theorem}[{\cite[Theorem 2]{BDKS17}}]
\label{thm:cgt}
Let $d,n$ be positive integers with $d \le n$, and $0 < \delta < 1$.  There is a distribution $\Dcal$ over $n$-by-$k$ Boolean matrices for $k = O(d(\log(n) + \log(1/\delta)))$ such that 
for any $x \in \01^n$ with $|x| \le d$, the string $x$ can be recovered from $[x^T R]_{>0}$ with probability at least $1-\delta$ when 
$R$ is chosen according to $\Dcal$.
\end{theorem}

We are now ready to give the proof of \cref{lem:recover}

\begin{proof}[Proof of \cref{lem:recover}]
By \cref{prop:boolean} we may assume that $A$ is Boolean and we have $(A \vee x) \circ (1-x)$ query access to $A$.
Let $k = |R|, \ell = |S|$ and $B = A(R,S)$ be the submatrix of interest.  Let $b \in \mathbb{N}^{k}$ where $b(i)$ is the number of ones in the $\ith$ row of $B$.
By \cref{lem:counting} applied with $\delta = 1/n^4$ and a union bound, there is a distribution $\mathcal{R}_{\mathrm count}$ over Boolean $\ell$-by-$t$ matrices with $t = O(\log(n)^2)$ such that from $[BK]_{>0}$ we can 
recover a good estimate of $b$ with probability at least $1-1/n^3$ when $K$ is drawn from $\mathcal{R}_{\mathrm count}$.  We add $1/n^3$ to our error bound and continue the proof assuming 
that we have a good estimate of $b$. As $R$ and $S$ are disjoint we can simulate the computation 
of $[BK]_{>0}$ with $O(\log(n)^2)$ many $(A \vee x) \circ (1-x)$ queries.

Next, for $i = 1, \ldots, \ceil{\log n}$ we bucket together rows of $B$ where the estimated number of non-zero entries is in the interval $(2^{i-1}, 2^{i}]$ into a set $G_i$.  As the estimate is good, 
for every $j \in G_i$ the number of non-zero entries in the $\jth$ row of $B$ is actually in the interval $(2^{i-3},2^{i+1}]$.
For each $i = 1, \ldots, \ceil{\log n}$ randomly sample $\frac{32 \ell \ln(n)}{2^{i}}$ elements of $S$ with replacement and let $H_i$ be the resulting set. 
By \cref{lem:sample}, with probability at least $1-1/n^3$ there is at least 1 and at most $O(\log{(n)})$ ones in each row of the 
submatrix $A(G_i,H_i)$.  We add $\log(n)/n^3$ to our error total and assume this is the case for all $i = 1, \ldots, \ceil{\log n}$.

For each $i = 1, \ldots, \ceil{\log n}$ we next learn the non-zero entries of $A(G_i,H_i)$ via \cref{thm:cgt}.  This theorem states that
for $t = O(\log(n)^2)$ there is a family of $|H_i|$-by-$t$ matrices $\Dcal$ such that from $[A(G_i,H_i) D]_{>0}$ we can learn the positions of all the ones of $A(G_i,H_i)$ 
with probability at least $1-\frac{1}{n^3}$, when $D$ is chosen from $\Dcal$.  Again by a union bound this computation will be successful for all $i$ except with probability $\log(n)/n^3$.  
As $R$ and $S$ are disjoint we can simulate the computation of $[A(G_i,H_i) D]_{>0}$ with $O(\log^2 n)$ many $(A \vee x) \circ (1-x)$ queries.  Over all $i$, the total number of queries 
is $O(\log(n)^3)$ and the error is at most $3 \log(n)/n^3 = O(1/n^2)$.  
\end{proof}

\begin{theorem}
\label{thm:con_master}
Let $G = (V,w)$ be an $n$-vertex weighted graph and $A$ its adjacency matrix.  There is a randomized algorithm that outputs a spanning forest of $G$ with probability at least $4/5$ after 
$O(\log(n)^4)$ many $[Av]_{>0} \circ (1-v)$ queries.
\end{theorem}

\begin{proof}
This follows from \cref{thm:main} and \cref{lem:recover}.
\end{proof}

\section{Applications}
In this section we look at consequences of \cref{thm:con_master} to algorithms with matrix-vector multiplication queries to the adjacency matrix, quantum algorithms with cut queries, and quantum 
algorithms with bipartite independent set queries.

\subsection{Matrix-vector multiplication queries}
Let $G = (V,E)$ be a simple $n$-vertex graph.  There are several different ways one can represent $G$ by a matrix, for example by its adjacency 
matrix $A \in \{0,1\}^{n \times n}$, or its \emph{signed vertex-edge incidence matrix} $A_{\pm} \in \{-1, 0, 1\}^{n \times \binom{n}{2}}$.
Sun et al.\ study the complexity of graph problems with matrix-vector multiplication queries, and ask the question if some matrix representations allow for much 
more efficient algorithms than others \cite[Question 4]{SWYZ19}.  They point out that the sketching algorithm for connectivity of Ahn, Guha, and McGregor \cite{AGM12} can be phrased 
as a matrix-vector multiplication query algorithm to the signed vertex-edge incidence matrix $A_\pm$.  Specifically, the AGM algorithm gives a randomized \emph{non-adaptive} algorithm 
that can output a spanning forest of an $n$-vertex graph after $O(\log^4(n))$ matrix-vector multiplication queries to $A_{\pm}$, where the queries are Boolean vectors. 
\footnote{The literal translation of the AGM algorithm uses $O(\log^3(n))$ many queries by vectors whose entries have $O(\log n)$ bits.  When translated to queries by Boolean vectors this 
results in $O(\log^4(n))$ queries.}  

In contrast, Sun et al.\ show that there is a bipartite graph $G$ such that when one is given matrix-vector multiplication query access \emph{on the right} to the \emph{bipartite} adjacency matrix 
of $G$, $\Omega(n/\log(n))$ many matrix-vector multiplication queries are needed to determine if $G$ is connected.  However, the instances they give can be solved by a \emph{single} matrix-vector 
multiplication query to the full adjacency matrix of $G$ (or a single matrix-vector multiplication query on the left to the bipartite adjacency matrix).

This leaves open the natural question if connectivity can be used to separate the matrix-vector multiplication query complexity of the adjacency matrix versus the signed vertex-edge incidence matrix 
representation.  As one can clearly simulate a query $[Ax]_{>0} \circ (1-x)$ in the master model by a matrix-vector multiplication query, our results show that connecitivity 
does not separate these models with the current state-of-the-art.

\begin{corollary}
\label{cor:mv}
Let $G = (V,w)$ be an $n$-vertex weighted graph.  There is an randomized algorithm that outputs a spanning forest of $G$ with probability at least $49/50$ after 
$O(\log(n)^4)$ matrix-vector multiplication queries to the adjacency matrix of $G$.
\end{corollary}

The $O(\log(n)^4)$ complexity matches the upper bound in the signed vertex-edge incidence matrix representation, however, the 
latter algorithm is non-adaptive while the adjacency matrix algorithm is not.  It is still possible that connectivity provides a large separation between the
non-adaptive adjacency matrix-vector multiplication and non-adaptive signed vertex-edge incidence matrix-vector multiplication models.

\subsection{Quantum cut queries}
Let $G = (V,w)$ be an $n$-vertex weighted graph.  In this section, we will assume that the edge weights are natural numbers in $\{0, \ldots, M-1\}$.  Let $A \in \{0, \ldots, M-1\}^{n \times n}$ be the adjacency 
matrix of $G$.  In a cut query, one can ask any $z \in \{0,1\}^n$ and receive the answer $(1-z)^T A z$.  A very similar query is a \emph{cross query}.  In a cross query one can query 
any $y,z \in \{0,1\}^n$ that are \emph{disjoint} (i.e.\ $y \circ z = \zero$) and receive the answer $y^T A z$.  It is clear that a cross query can simulate a cut query.  One can also simulate a cross query with 
3 cut queries because for disjoint $y,z$ 
\[
y^T A z = \frac{1}{2} \left((1-y)^T A y + (1-z)^T A z - (1-(y+z))^T A (y+z) \right) \enspace .
\]
Because of this constant factor equivalence we will make use of cross queries when it is more convenient.  

Lee, Santha, and Zhang \cite{LSZ21} give a quantum algorithm to find a spanning forest of a simple $n$-vertex graph 
after $O(\log^8(n))$ cut queries.  A key observation they make is that a quantum algorithm can efficiently simulate a 
restricted version of a matrix-vector multiplication query to the adjacency matrix.  This follows from the 
following lemma, which is an adaptation of the Bernstein-Vazirani algorithm \cite{BV97}.

\begin{lemma}[{\cite[Lemma 9]{LSZ21}}]\label{lem:exactqalgo}
Let $x \in \{0,1,\ldots, K-1\}^n$ and suppose we have access to an oracle that returns $\sum_{i \in S}x_i \mod K$ for any $S \subseteq [n]$. Then there exists a quantum algorithm that learns $x$ with $O(\ceil{\log(K)})$ queries without any error.
\end{lemma}

\begin{corollary}[cf.\ {\cite[Lemma 11]{LSZ21}}]
\label{cor:qcut_recover}
Let $G = (V,w)$ be an $n$-vertex weighted graph with integer weights in $\{0,1, \ldots, M-1\}$, and let $A$ be its adjacency matrix.
There is a quantum algorithm that for any $z \in \{0,1\}^n$ perfectly computes $x = Az \circ (1-z)$ with $O(\log(Mn))$ cut queries to $A$.
\end{corollary}

\begin{proof}
Let $x = Az \circ (1-z)$.  The entries of $x$ are in $\{0,\ldots, n(M-1)\}$, thus we can work modulo $K=2Mn$ and preserve all entries of $x$.  Moreover, 
$x$ is zero wherever $z$ is one.  Thus for any $y \in \{0,1\}^n$ we have $y^T x = (y \circ (1-z))^T x = (y \circ (1-z))^T A z$, which can be computed with one cross query.
Thus we can apply \cref{lem:exactqalgo} to learn $x$ perfectly with $O(\log(Mn))$ many cut queries to $A$.
\end{proof}

Using \cref{cor:qcut_recover}, classical algorithms using $Az \circ (1-z)$ queries give rise to quantum algorithm using cut queries. Thus we can apply \cref{thm:con_master} to obtain the 
following quantitative improvement of the result of \cite{LSZ21}.

\begin{theorem}
\label{thm:qcut}
Given cut query access to an $n$-vertex graph $G$ with integer weights in $\{0,1,\ldots, M-1\}$ there is a quantum algorithm making $O(\log^4(n) \log(Mn))$ queries that 
outputs a spanning forest of $G$ with at least probability $49/50$.
\end{theorem}

\begin{proof}
The proof follows from \cref{cor:qcut_recover} and \cref{thm:con_master}.
\end{proof}

\subsection{Quantum bipartite independent set queries}
Let $G = (V,E)$ be a simple graph and $A$ its adjacency matrix.  In a bipartite independent set query, one can query any \emph{disjoint} $y,z \in \{0,1\}^n$ and 
receive the answer $[y^T A z]_{> 0} \in \{0,1\}$.  A query in the master model can be simulated by a quantum algorithm making $\tilde O(\sqrt{n})$ bipartite independent 
set queries.  The key to this result is the following theorem of Belovs about combinatorial group testing with quantum algorithms.

\begin{theorem}[Belovs \cite{Belovs15}]
\label{thm:belovs}
Let $x \in \{0,1\}^n$.  There is a quantum algorithm that outputs $x$ with probability at least $2/3$ after $O(\sqrt{n})$ OR queries to $x$.  
\end{theorem}

\begin{corollary}
\label{cor:bipartite}
Let $G = (V,E)$ be an $n$-vertex simple graph, and let $A$ be its adjacency matrix.
There is a quantum algorithm that for any $z \in \{0,1\}^n$ computes $Az \circ (1-z)$ with probability at least $1-1/n^3$ after $O(\sqrt{n} \log(n))$ bipartite independent set queries to $G$.
\end{corollary}

\begin{proof}
Let $x = Az \circ (1-z)$.  For any $y \in \{0,1\}^n$ we have $y^T x = (y \circ (1-z))^T x = (y \circ (1-z)) A z$, thus $[y^T x]_{> 0}$ can be computed with a single bipartite independent 
set query to $G$ as $y \circ (1-z)$ and $z$ are disjoint.  Therefore by \cref{thm:belovs} and error reduction, with $O(\sqrt{n}\log(n))$ bipartite independent set 
queries we can compute $x$ with probability at least $1-1/n^3$.
\end{proof}

\begin{theorem}
Let $G = (V,E)$ be an $n$-vertex simple graph.  There is a quantum algorithm which outputs a spanning forest of $G$ with probability at least $4/5$ after 
$O(\sqrt{n} \log(n)^5)$ bipartite independent set queries to $G$.
\end{theorem}

\begin{proof}
Let $A$ be the adjacency matrix of $G$.
By \cref{thm:con_master} there is an algorithm that outputs a spanning forest of $G$ with probability at least $49/50$ after $O(\log(n)^4)$ many $[Av]_{> 0} \circ (1-v)$ 
queries.  We can implement an $[Av]_{> 0} \circ (1-v)$ with error at most $1/n^3$ by a quantum algorithm making $O(\sqrt{n}\log(n))$ bipartite independent set 
queries by \cref{cor:bipartite}.  The probability any of the $O(\log(n)^4)$ many $[Av]_{> 0} \circ (1-v)$ queries is computed incorrectly is at most $O(\log(n)^4/n^3)$.  Thus 
there is a quantum algorithm with success probability at least $4/5$ that outputs a spanning forest of $G$ after $O(\sqrt{n} \log(n)^5)$ bipartite independent set queries.
\end{proof}

\section{Lower bound for linear queries}
On input an $n$-vertex graph $G = (V,w)$, a linear query algorithm can ask any $x \in \R^{\binom{n}{2}}$ and receive the 
answer $\braket{w}{x}$.  In this section we show that any deterministic, or even zero-error randomized, algorithm for connectivity must make $\Omega(n)$ linear queries.  
As far as we are aware, this is the first lower bound of any kind for connectivity in the unrestricted linear query model.  

The reason lower bounds against linear queries are difficult to show is because the answer to a query can have an unbounded 
number of bits.  The ``hard'' inputs for a linear query lower bound cannot be limited to simple graphs, as given the promise that a graph is simple 
it can be learned with a single linear query, by querying the vector of powers of 2 of the appropriate dimension.

Most global query lower bounds for connectivity are derived from communication complexity.
It is known that the communication complexity of connectivity is 
$\Theta(n \log n)$ in the deterministic case \cite{HMT88} and $\Omega(n)$ in the bounded-error randomized case \cite{BFS86}.  Moreover, 
the family of instances used to show these lower bounds are all simple graphs.
If the answer to a global query on a simple graph can be communicated with at most $b$ bits, then these communication results imply
$\Omega(n \log(n)/b)$ and $\Omega(n/b)$ deterministic and randomized lower bounds on the query complexity of connectivity in this model, respectively.
In particular, this method gives an $\Omega(n)$ deterministic and $\Omega(n/\log n)$ randomized lower bound on the cut query complexity 
of connectivity.  This approach, however, cannot show \emph{any} lower bound against linear queries.

A lower bound technique which goes beyond considering simple graphs is the \emph{cut dimension} \cite{GPRW20}.  This technique 
was originally developed by Graur et al.\  for showing cut query lower bounds against deterministic algorithms solving the minimum cut problem.  In 
\cite{LLSZ21} it was observed that the lower bound also applies to the linear query model, and a strengthening of the technique was 
given called the \emph{$\ell_1$-approximate cut dimension}.  Here we observe that the $\ell_1$-approximate cut dimension characterizes, up to 
an additive $+1$,
a well known query complexity lower bound technique applied to the minimum cut problem, the \emph{certificate complexity}.  We further adapt this certificate complexity 
technique to the connectivity problem, allowing us to show an $\Omega(n)$ lower bound for deterministic linear query 
algorithms solving connectivity.  By its nature, certificate complexity also lower bounds the query complexity of zero-error randomized 
algorithms, thus we get a lower bound of $\Omega(n)$ for connectivity in this model as well.

We first define the linear query certificate complexity in \cref{sec:def_cert} and show that it is equivalent to the $\ell_1$-approximate cut dimension of \cite{LLSZ21}.  
Then in \cref{sec:apply_cert} we apply this method to show an $\Omega(n)$ lower bound on the linear query certificate complexity of connectivity.

\subsection{Certificate complexity}
\label{sec:def_cert}
While our main focus is connectivity, the ideas here apply to the minimum cut problem as well, so we treat both cases.
We will need some definitions relating to cuts.  Let $G = (V,w)$ be a weighted graph.  For $\emptyset \ne S \subsetneq V$ the cut $\Delta_G(S)$
corresponding to $S$ is the set of edges of $G$ with exactly one endpoint in $S$.  When the graph is clear from context we will drop the subscript.  
We call $S$ a \emph{shore} of the cut.  Every cut has two shores.
When we wish to speak about a unique shore, for example when enumerating cuts, we will take the shore not containing some distinguished vertex $v_0 \in V$.  
The weight of a cut $w(\Delta_G(S))$ is $\sum_{e \in \Delta_G(S)} w(e)$.  We let $\lambda(G)$ denote the weight of a minimum cut in $G$.

\begin{definition}[$\con$ and $\mincut$]
The input in the $\con$ and $\mincut$ problems is an $n$-vertex weighted undirected graph $G = (V,w)$.  In the $\con$ problem 
the goal is to output if $\lambda(G) > 0$ or not.  In the $\mincut$ problem the goal is to output $\lambda(G)$.  
\end{definition}

A deterministic algorithm correctly solves the $\mincut$ problem if it outputs $\lambda(G)$ for every $n$-vertex input graph $G$.  
We let $D_{\lin}(\mincut)$ denote the minimum, over all deterministic linear query algorithms $\Acal$ that correctly solve $\mincut$, 
of the maximum over all $n$-vertex input graphs $G = (V,w)$ of the number of linear queries made by $\Acal$ on $G$.  The deterministic linear 
query complexity of connectivity, $D_{\lin}(\con)$, is defined analogously.

We will also consider zero-error randomized linear query algorithms.  A zero-error randomized linear query algorithm for $\mincut$ is a probability distribution over 
deterministic linear query algorithms that correctly solve $\mincut$.  The cost of a zero-error randomized algorithm $\Acal$ on input $G = (V,w)$ is the 
expected number of queries made by $\Acal$.  We let $R_{0,\lin}(\mincut)$ denote the minimum over all zero-error randomized algorithms $\Acal$ for $\mincut$ 
of the maximum over all $G$ of the cost of $\Acal$ on $G$.  $R_{0,\lin}(\con)$ is defined analogously.  Clearly $R_{0,\lin}(\mincut) \le D_{\lin}(\mincut)$ and 
$R_{0,\lin}(\con) \le D_{\lin}(\con)$

We will investigate a lower bound technique called certificate complexity.
\begin{definition}[Certificate complexity]
Let $G = (V,w)$ be an $n$-vertex graph.  
The minimum-cut linear-query certificate complexity of $G$, denoted $\mincutcert(G)$, is the minimum $k$ 
such that there is a matrix $A \in \R^{k \times \binom{n}{2}}$ with the property that, for any graph $G' = (V,w')$, if 
$Aw = Aw'$ then $\lambda(G) = \lambda(G')$.

The connectivity linear-query certificate complexity, denoted $\concert(G)$, is the minimum $k$ 
such that there is a matrix $A \in \R^{k \times \binom{n}{2}}$ with the property that, for any graph $G' = (V,w')$, if 
$Aw = Aw'$ then $\lambda(G') > 0$ iff $\lambda(G) > 0$.  
\end{definition}

It is a standard fact that certificate complexity is a lower bound on zero-error randomized query complexity.
We reproduce the proof here for completeness.
\begin{lemma}
\label{lem:R0}
For any $n$-vertex graph $G=(V,w)$
\begin{align*}
R_{0,\lin}(\mincut) &\ge \mincutcert(G) \\
R_{0,\lin}(\con) &\ge \concert(G) \enspace .
\end{align*}
\end{lemma}

\begin{proof}
We treat the minimum cut case, the connectivity case follows similarly. 

First consider a deterministic linear query algorithm for minimum cut on input $G$.  This algorithm must make 
at least $\mincutcert(G)$ many queries.  If not, there is a graph $G'$ which agrees with $G$ on all the queries
 but such that $\lambda(G') \ne \lambda(G)$.  As the algorithm does not distinguish $G$ from $G'$ it 
cannot answer correctly.

Consider now a zero-error randomized linear query algorithm $\Acal$ for minimum cut.  Each deterministic algorithm in the support of $\Acal$ correctly 
solves $\mincut$ and so must make at least $\mincutcert(G)$ many 
queries on input $G$.  Thus expected number of queries made by $\Acal$ on input $G$ is at least $\mincutcert(G)$ as well.
\end{proof}

It is useful to think of a certificate in two parts.  Let $A \in \R^{k \times \binom{n}{2}}$ be a minimum cut certificate for $G$.  
Then for any $G' = (V,w')$ with 
$Aw = Aw'$ it must hold that $\lambda(G) = \lambda(G')$.  The condition that $\lambda(G') \le \lambda(G)$ can be certified with 
a single query.  If $S$ is the shore of a minimum cut in $G$ then we can query $\chi_S \in \{0,1\}^{\binom{n}{2}}$, the 
characteristic vector of $\Delta_{K_n}(S)$ where $K_n$ is the complete graph.  Then $\braket{\chi_S}{w} = \lambda(G)$, and 
including $\chi_S$ as a row of $A$ guarantees that if $Aw = Aw'$ then $\lambda(G') \le \lambda(G)$.

The more challenging problem is certifying that the minimum cut of any $G' = (V,w')$ with $Aw = Aw'$ is at least $\lambda(G)$.  
We single out a more general version of this lower bound certification problem, which will be useful to treat the minimum cut and 
connectivity cases together.
\begin{definition}
Let $G = (V,w)$ be a graph and let $0 \le \tau \le \lambda(G)$ be a parameter.  The $\taucerta$ of $G$ is the least $k$ such that there is 
a matrix $A \in \R^{k \times \binom{n}{2}}$ such that for any $G' = (V,w')$ with $Aw = Aw'$ it holds that $\lambda(G') \ge \tau$.
\end{definition}

\begin{lemma}
\label{lem:gencert}
Let $G = (V,w)$ be a graph.  Then
\[
\textrm{at-least-}\lambda(G)\textrm{-cert}(G) \le \mincutcert(G) \le \textrm{at-least-}\lambda(G)\textrm{-cert}(G) +1 \enspace .
\]
If $G$ is connected then $\concert(G) = \inf_{\tau > 0} \taucerta(G)$.
\end{lemma}

\begin{proof}
If $A$ is a $\mincutcert$ for $G$ then in particular it certifies that $\lambda(G') \ge \lambda(G)$ for any $G' = (V,w')$ with $Aw = Aw'$.  This 
shows the first lower bound.  For the other direction, suppose $A$ is an $\textrm{at-least-}\lambda(G)\textrm{-cert}$ for $G$.  Let $S$ be the 
shore of a minimum cut in $G$ and $\chi_S \in \{0,1\}^{\binom{n}{2}}$, the characteristic vector of the cut corresponding to $S$ in the complete graph.
Then adjoining $\chi_S$ as an additional row to $A$ creates a minimum cut certificate for $G$.

Now we show the second part of the lemma.  Let $G$ be a connected graph.  For any $0 < \tau \le \lambda(G)$, if $A$ is an $\taucerta$ certificate for $G$
then it holds that for any $G' = (V,w')$ with $Aw = Aw'$ that $\lambda(G') \ge \tau > 0$ so $G'$ is also connected.  This shows that 
$\concert(G) \le  \inf_{\tau > 0} \taucerta(G)$.

For the other direction, let $A$ be a connectivity certificate for $G$.  We claim that $A$ is also a $\taucerta$ for some $\tau >0$.  
For $\emptyset \ne S \subsetneq V$ let $\chi_S \in \{0,1\}^{\binom{n}{2}}$ be the characteristic vector of $\Delta_{K_n}(S)$, where 
$K_n$ is the complete graph on $n$ vertices.  Define 
\begin{equation*}
\begin{aligned}
\iota(S) = & \underset{w'}{\text{minimize}}
& & \braket{\chi_S}{w'} \\
& \text{subject to}
& & Aw = Aw'  \\
& & & w'  \ge \zero \enspace .
\end{aligned}
\end{equation*}
As $\iota(S)$ is defined by a linear program, the minimum is achieved, and is strictly positive as $A$ is a connectivity certificate.  Thus letting 
$\tau = \min_{\emptyset \ne S \subsetneq V} \iota(S)$ we see that $A$ is an $\taucerta$.  This shows $\inf_{\tau > 0} \taucerta(G) \le \concert(G)$.
\end{proof}

We now proceed to give a linear-algebraic characterization 
of these certificate complexities.  First a definition.

\begin{definition}[universal cut-edge incidence matrix]
Let $K_n$ be the complete graph on vertex set $\{1, \ldots, n\}$.  The universal cut-edge incidence matrix $M_n$ is a Boolean $(2^{n-1}-1)$-by-$\binom{n}{2}$ matrix 
with rows indexed by non-empty sets $S$ with $1 \not \in S$ and columns indexed by edges $e$ and $M_n(S,e)=1$ iff $e \in \Delta_{K_n}(S)$.  
\end{definition}

\begin{definition}
\label{def:cutrank}
Let $G = (V,w)$ be a graph on $n$ vertices and $M_n$ the universal cut-edge incidence matrix.  
For $\tau \in \R$ with $0 \le \tau \le \lambda(G)$ define
\begin{equation*}
\begin{aligned}
\tau \cutrk(G) = \ & \underset{X}{\text{minimize}}
& & \rank(X) \\
& \text{subject to}
& & X \leq M_n \\
& & & X w \ge \tau \enspace .
\end{aligned}
\end{equation*}
\end{definition}

\begin{remark}
The $\lambda(G) \cutrk(G)$ is equivalent to the $\ell_1$ approximate
cut dimension defined in \cite{LLSZ21}.  

We have defined the cut rank where 
$X$ is a matrix with $\binom{n}{2}$ columns as this will be more convenient in the 
proof of \cref{thm:rkequiv}.  If $G$ has $m$ edges, one could instead restrict
$M_n$ and $X$ to matrices with $m$ columns, corresponding to the edges of $G$.  

Indeed, let $\widehat M_n$ be the $(2^{n-1}-1)$-by-$m$ matrix that is $M_n$ with columns restricted 
to edges of $G$ and similarly $\hat w \in \R^m$ be $w$ restricted to edges of $G$.  One can see that the 
program
\begin{equation*}
\begin{aligned}
& \underset{Y}{\text{minimize}}
& & \rank(Y) \\
& \text{subject to}
& & Y \leq \widehat M_n \\
& & & Y \hat w \ge \tau 
\end{aligned}
\end{equation*}
has value equal to $\tau \cutrk(G)$.  It is clear that $\tau \cutrk(G)$ is at most the value of this program as any feasible $Y$ 
can be turned into a feasible $X$ by populating the additional columns with zeros.  

For the other direction, let let $X^*$ realize $\tau \cutrk(G)$ and let $Y^*$ be $X^*$ with columns corresponding to non-edges of $G$ deleted.  
Then $\rk(Y^*) \le \rk(X^*)$, $Y^* \le \widehat M_n$ and $Y^* \hat w \ge \tau$ because $w$ is zero on all entries corresponding to the deleted columns.
When we are actually proving lower bounds in \cref{sec:apply_cert}, it will be more convenient to use this formulation where columns are restricted to edges of $G$.
\end{remark}

\begin{theorem}
\label{thm:rkequiv}
For any $n$-vertex graph $G = (V,w)$ and $0 \le \tau \le \lambda(G)$ we have $\taucerta(G) = \tau \cutrk(G)$.  
\end{theorem}

\begin{proof}
We first show $\taucerta(G) \ge \tau \cutrk(G)$.  Suppose the
$\taucerta$ certificate complexity of $G$ is $k$ and let $A$ be a $k$-by-$\binom{n}{2}$ matrix realizing this. 
Let $G' = (V,w')$ be such that $Aw = Aw'$.  Let us write $w' = w-z \ge \zero$ where $Az = \zero$.

Consider a shore $\emptyset \ne S \subsetneq V$.  Let $\chi_S \in \{0,1\}^{\binom{n}{2}}$ be the characteristic vector of 
$\Delta_{K_n}(S)$, where $K_n$ is the complete graph.  Thus $w(\Delta_G(S)) = \braket{\chi_S}{w}$ and 
\[
w'(\Delta_{G'}(S)) = \braket{\chi_S}{w-z} = w(\Delta_G(S)) - \braket{\chi_S}{z} \enspace .
\]
As $A$ is an $\taucerta$ certificate, we must have $w(\Delta_G(S)) - \braket{\chi_S}{z} \ge \tau$ which means 
$\braket{\chi_S}{z} \le w(\Delta_G(S)) - \tau$.  

Consider the optimization problem
\begin{equation*}
\begin{aligned}
\alpha(S) = & \underset{z}{\text{maximize}}
& & \braket{\chi_S}{z} \\
& \text{subject to}
& & Az = \zero \\
& & & w - z \ge \zero \enspace .
\end{aligned}
\end{equation*}
As we have argued, $\alpha(S) \le w(\Delta_G(S)) - \tau$.
The dual of this problem is 
\begin{equation*}
\begin{aligned}
\beta(S) = & \underset{v}{\text{minimize}}
& & \braket{w}{\chi_S-A^Tv} \\
& \text{subject to}
& & \chi_S-A^T v \ge \zero \enspace.
\end{aligned}
\end{equation*}
The primal is feasible with $z = \zero$ so we have strong duality and $\alpha(S) = \beta(S)$.  This means 
there exists a vector $X_S$ in the row space of $A$ with $\chi_S \ge X_S$ and 
$\braket{w}{\chi_S-X_S} \le w(\Delta_G(S)) - \tau$, which implies $\braket{X_S}{w} \ge \tau$.  

As $S$ was arbitrary, this holds for every shore.  Package the vectors $X_S$ as the rows of a matrix $X$.  
This matrix satisfies $X \le M_n$, since $X_S \le \chi_S$, and $Xw \ge \tau$.  Further, $\rk(X) \le k$ as every row of $X$ is in the row space of $A$.

Now we show $\tau \cutrk(G) \ge \taucerta(G)$.  Let $X$ be a matrix realizing $\tau \cutrk(G)$ and let $k = \rk(X)$.
Let $A$ be a matrix whose rows are a basis for the row space of $X$, and so $\rk(A) = k$.  
To show that $A$ is an at-least-$\tau$ certificate for $G$ it suffices to show that $\alpha(S) \le w(\Delta_G(S)) - \tau$ for 
every shore $S$.

Let $X_S$ be the row of $X$ corresponding to shore $S$.  We have that $X_S \le \chi_S$ and that $X_S$ is in the row space of $A$.  
Thus $\alpha(S) = \beta(S) \le \braket{w}{\chi_S - X_S} \le w(\Delta_G(S)) - \tau$.  This completes the proof.
\end{proof}

\begin{corollary}
\label{cor:rk}
Let $G = (V,w)$ be a graph.  Then
\[
\lambda(G)\cutrk(G) \le \mincutcert(G) \le \lambda(G) \cutrk(G) +1 \enspace .
\]
If $G$ is connected then $\concert(G) = \inf_{\tau > 0} \tau \cutrk(G)$.
\end{corollary}

\subsection{Application to connectivity}
\label{sec:apply_cert}
\begin{theorem}
\label{thm:cycle_cert}
Let $n$ be even and $C_n = ([n],w)$ be a cycle graph on $n$ vertices where all edge weights are $1$.  Then $\concert(C_n) \ge n/4$.
\end{theorem}

\begin{proof}
We will show that $\inf_{\tau > 0} \tau \cutrk(C_n) \ge n/4$ which will give the theorem by \cref{cor:rk}.

Let $w$ be the weight vector of $C_n$ and let $X$ be a $(2^{n-1}-1)$-by-$\binom{n}{2}$ matrix satisfying $X \le M_n$ and $Xw > \zero$.  
We want to show that $\rk(X) \ge n/4$.  Let $Y$ be a $n/2$-by-$n$ submatrix of $X$ where the rows are restricted to the singleton 
shores $\{2\},\{4\},\ldots, \{n\}$ at even numbered vertices, and the columns are restricted to the edges $e_{12}, e_{23}, \ldots, e_{n1}$ of $C_n$.  To show a lower bound 
on the rank of $X$ it suffices to show a lower bound on the rank of $Y$.

The cuts $\Delta(\{2\}), \Delta(\{4\}), \ldots, \Delta(\{n\})$ partition the edges of $C_n$, therefore every pair of rows of $Y$ are disjoint.  
Note that $Xw$ only depends on the columns of $X$ corresponding to edges of $C_n$.  These are the columns we have restricted to in $Y$, 
therefore since $Xw > \zero$, every row sum of $Y$ must be positive.  Also note that entries of $Y$ in the row corresponding to shore $\{2i\}$ can 
only potentially be positive in the columns labeled by $e_{2i-1,2i}$ and $e_{2i,2i+1}$, where addition is taken modulo $n$.

We further modify $Y$ to a square $n/2$-by-$n/2$ matrix $Y'$ by summing together columns $2i-1$ and $2i$ for $i=1, \ldots, n/2$. 
Note that $\rank(Y') \le \rank(Y)$ and now $Y'$ has strictly positive entries on the diagonal, and must be non-positive everywhere else.  
Next we multiply each row by the appropriate number so that the diagonal entry becomes $1$.  This preserves the property that in each row the sum of the positive entries (which is now $1$) 
is strictly greater than the sum of the negative entries, and does not change the rank.  Thus in each row the $\ell_1$ norm of the off-diagonal entries is at most $1$, and by the Gershgorin circle theorem all eigenvalues of $Y'$ are at most $2$.  
On the other hand, the trace of $Y'$ is $n/2$, thus the rank must be at least $n/4$.
\end{proof}

\begin{corollary}
\label{cor:lin_low}
$R_{0,\lin}(\con) \ge n/4.$
\end{corollary}

\begin{proof}
Follows by \cref{lem:R0} and \cref{thm:cycle_cert}.
\end{proof}

\section{Open problems}
We conclude with some open questions.
\begin{enumerate}
  \item We have shown that connectivity has an efficient algorithm with matrix-vector multiplication queries to the adjacency matrix.  It remains an interesting question 
  to find an example of a graph problem that can be solved much more efficiently with matrix-vector multiplication queries to the signed vertex-edge incidence matrix 
  than with matrix-vector multiplication queries to the adjacency matrix.  Sun et al.\ \cite{SWYZ19} show that one can find a spectral sparsifier of a graph with $\polylog(n)$ 
  matrix-vector multiplication queries to the signed vertex-edge incidence matrix.  This means that one can solve the problem of determining if the edge connectivity 
  of a simple graph is at least $2k$ or at most $k$ with $\polylog(n)$  matrix-vector multiplication queries to the signed vertex-edge incidence matrix.  It is an interesting 
  open question if this can also be done with $\polylog(n)$ matrix-vector multiplication queries to the adjacency matrix.
  \item What is the randomized \emph{non-adaptive} complexity of connectivity with matrix-vector multiplication queries to the adjacency matrix?  The $O(\log^4(n))$ query 
  algorithm in the signed vertex-edge incidence matrix model is non-adaptive, but we do not see how to design a non-adaptive algorithm making $\polylog(n)$ matrix-vector multiplication queries to the adjacency 
  matrix.
  \item How large can the minimum-cut linear-query certificate complexity of a graph be? 
  \item What is the bounded-error randomized linear query complexity of connectivity?  It seems that new ideas are needed to show lower bounds for the bounded-error model.
\end{enumerate}

\section*{Acknowledgements}
TL would like to thank Simon Apers, Yuval Efron, Pawel Gawrychowski, Sagnik Mukhopadhyay, Danupon Nanongkai, and Miklos Santha for many conversations
which contributed ideas to this paper.  TL also thanks Deeparnab Chakrabarty for answering questions about \cite{ACK21}.  TL is supported in part by the Australian 
Research Council Grant No: DP200100950.

%\bibliographystyle{alpha}
%\bibliography{sub.bib}

\newcommand{\etalchar}[1]{$^{#1}$}

\end{document}